\documentclass[11pt]{article}
\pdfoutput=1
\AddToHook{package/amsmath/before}{%

}

\usepackage{amsmath}
\usepackage{amsthm}
\usepackage{amssymb}
\usepackage{mathrsfs}
\usepackage[T1]{fontenc}
\usepackage[english]{babel}
\usepackage{varioref}
\usepackage{calc}
\usepackage{ifthen}
\usepackage{enumerate}
\usepackage{paralist}
\usepackage{graphics} 
\usepackage{tikz}
\usepackage{url}
\usetikzlibrary{calc,through,backgrounds,arrows}
\usepackage{fullpage}
\usepackage{authblk}

\usepackage{color}
\definecolor{webblue}{rgb}{0, 0, 1.0}  
\definecolor{webgreen}{rgb}{0,1.0,0} 
\definecolor{webred}{rgb}{1.0, 0, 0}   

\definecolor{lily1}{rgb}{0.5,0,0.5}
\definecolor{lily2}{rgb}{0.75,0,0.25}
\definecolor{lily3}{rgb}{0.25,0,0.75}
\definecolor{green1}{rgb}{0.25,0.25,0}
\definecolor{green2}{rgb}{0.5,0.5,0}
\definecolor{green3}{rgb}{0.75,0.75,0}

\newcommand{\N}{\mathbb{N}}

\newcommand{\cA}{\mathcal{A}}
\newcommand{\cE}{\mathcal{E}}
\newcommand{\cN}{\mathcal{N}}
\newcommand{\cU}{\mathcal{U}}
\newcommand{\cV}{\mathcal{V}}

\newtheorem{theorem}{Theorem}

\newtheorem{lemma}[theorem]{Lemma}

\newtheorem{definition}[theorem]{Definition}

\newtheorem{example}[theorem]{Example}

\providecommand{\keywords}[1]
{
  \small	
  \textbf{Keywords} #1
}

\providecommand{\subclass}[1]
{
  \small	
  \textbf{Mathematics Subject Classification (2000)} #1
}

\begin{document}

\title{Capacity of an infinite family of networks related to the diamond network for fixed alphabet sizes}

\author{Sascha Kurz}
\affil{Mathematisches Institut, Universit\"at Bayreuth, D-95440 Bayreuth, Germany, sascha.kurz@uni-bayreuth.de}

\date{}

\maketitle

\begin{abstract}
We consider the problem of error correction in a network where the errors can occur only 
on a proper subset of the network edges. For a generalization of the so-called Diamond 
Network we consider lower and upper bounds for the network's ($1$-shot) capacity for fixed 
alphabet sizes.

\medskip

\noindent

\noindent
\subclass{94B65}
\end{abstract}


\section{Introduction}\label{sec:intro}
The correction of errors introduced by noise or adversaries in networks have been studied in a number of papers, see 
e.g.\ \cite{cai2006network,yeung2006network}. Here we assume that the adversary can manipulate a subset of the edges of a given network. 
If the adversary is unrestricted, i.e., all edges can be manipulated, then it is well known that the $1$-shot capacity 
of the network can be achieved, 
see e.g.\ \cite{koetter2008coding}. 
In \cite{beemer2022network,beemer2021curious} specific networks are considered where the actual $1$-shot capacity cannot be achieved using
linear operations at intermediate nodes, so that network decoding is required in order to achieve the capacity. For a unified combinatorial 
treatment of adversarial network channels we refer to \cite{ravagnani2018adversarial}.

The aim of this paper is to study the capacity of a generalization of the diamond network, see \cite{beemer2022network,beemer2021curious} for the 
diamond network and Figure~\ref{fig_generalized_diamond_network} for the generalized case. This family of networks is one of five infinite 
families studied in \cite{beemer2022network_arxiv}, more precisely family~B. The dashed edges $e_0,\dots,e_{s+1}$ can be manipulated by 
the adversary and the solid edges $f_0,\dots,f_s$ are those that cannot be manipulated. Node $S$ is the source, $T$ the terminal, and $V_1,V_2$ 
are the intermediate nodes. We assume that the adversary can manipulate exactly one dashed edge and we are interested in the $1$-shot capacity, i.e., 
where the edges of the network are just used once. 

\begin{figure}
\begin{center}
  \begin{tikzpicture}[line width=1pt, scale=0.5]
  \draw[black] (0,2) circle (8mm);	
  \draw[black] (5,0) circle (8mm);
  \draw[black] (5,4) circle (8mm);
  \draw[black] (10,2) circle (8mm);
  \node at (0,2) {$S$};	
  \node at (5,0) {$V_2$};
  \node at (5,4) {$V_1$};
  \node at (10,2) {$T$};
  \node at (2.5,3.3) {$e_0$};
  \node at (7.5,3.45) {$f_0$};
  \node at (2.5,2.2) {$e_1$};
  \node at (2.5,0.7) {$\vdots$};
  \node at (2.5,1.5) {$\vdots$};
  \node at (2.5,-0.4) {$e_{s+1}$};
  \draw[dotted, ->] (0.75,2.3) -- (4.1,3.64);
  \draw[dotted, ->] (0.75,1.7) -- (4.1,0.36);	
  \draw[dotted, ->] (0.75,1.8) .. controls (2.5,2.1) .. (4.1,0.56);
  \draw[dotted, ->] (0.75,1.6) .. controls (2,-0.1) .. (4.1,0.20);
  \draw[->] (5.75,3.7) -- (9.1,2.36);	
  \draw[->] (5.75,-0.2) .. controls (7,1.9) .. (9.1,1.76);
  \draw[->] (5.75,-0.4) .. controls (8,-0.7) .. (9.1,1.30); 	
  \node at (7.5,2.2) {$f_1$};
  \node at (7.5,1) {$\vdots$};
  \node at (7.5,-0.2) {$f_{s}$};	
  \end{tikzpicture}
  \caption{The network $\cN_s$.}
  \label{fig_generalized_diamond_network}
\end{center}
\end{figure}

The possible mappings from the input to the output space of the nodes are numerous. So, after stating the 
necessary preliminaries in Section~\ref{sec_preliminaries}, we study the structure of {\lq\lq}good{\rq\rq} 
network codes in the single-error correction setting for generalized diamond networks $\cN_s$ in 
Section~\ref{sec_structure}. Our main structure result, see Theorem~\ref{thm_cover_problem}, states that we 
may assume that node $V_1$ just forwards its received information and node $V_2$ performs a partial error-correction 
and sends sets of codeword candidates as states, so that the design of an optimal network code becomes a covering problem 
for sets.  Based on these insights we determine lower and upper bounds for the $1$-shot 
capacity of $\cN_s$ in Sections~\ref{sec_ilp} and \ref{sec_bounds}, respectively.  
We close with a brief conclusion in Section~\ref{sec_conclusion}. 

While the family of networks $\cN_s$ seems to be rather specific and non-general, there are a few reasons to 
study this specific class in detail. It is one of three families of so-called $2$-level networks whose 
capacity couldn't be determined in \cite{beemer2022network_arxiv}. If we remove vertex $V_1$, then we are in 
the classical situation of coding along a single edge, so that this class is somehow the smallest possible 
example where we can see network (coding) effects. Even though $\cN_s$ is very well structured, the exact determination 
of its $1$-shot capacity for a given parameter $s$ and a given alphabet size seems to be a combinatorial 
challenge. In \cite{beemer2022network_arxiv} reductions for general networks to $3$-level networks and then 
to $2$-level networks were presented, so that it makes sense to study the smaller and more structured 
networks first. Nevertheless there is some hope that some of the presented techniques may generalize to 
different networks too.

\section{Preliminaries}
\label{sec_preliminaries}
In this section we formally define communication networks, network codes, and the channels they induce cf.~ \cite{beemer2022network_arxiv,ravagnani2018adversarial}. 
Due to a more restricted setting we use slightly different notation and refer the interested reader to \cite{ravagnani2018adversarial} for more details.  

\begin{definition}
  A \emph{(single-source communication) network} is a $4$-tuple $\cN=(\cV,\cE,S,\mathbf{T})$, where $(\cV,\cE)$ is 
  a finite, directed, acyclic multigraph $S\in\cV$ is the \emph{source}, and $\emptyset\neq \mathbf{T}\subseteq
  \cV\backslash\{S\}$ is the set of \emph{terminals}. Additionally we assume that there exists a directed path 
  from $S$ to any $T\in\mathbf{T}$ and for every $V\in\cV\backslash(\{S\}\cup\mathbf{T})$ there exist 
  directed paths from $S$ to $V$ and from $V$ to some $T\in\mathbf{T}$.    
\end{definition}

The elements of $\cV$ are called \emph{nodes}, the elements of $\cE$ are called \emph{edges}, and the 
elements of $\cV\backslash(\{S\}\cup\mathbf{T})$ are called \emph{intermediate nodes}. The set of incoming 
and outgoing edges for a node $V\in\cV$ are denoted by $in(V)$ and $out(V)$, respectively. 
As alphabet we choose an arbitrary finite set $\cA$ with at least two elements, e.g.\ $\cA=\left\{0,1,\dots,|\cA|-1\right\}$.

In order to describe the transmission across the network we need functions $F_V\colon \cA^{in(V)}\to\cA^{out(V)}$ for each intermediate
node $V$. We write $A^B:=\left\{ \varphi\colon B\to A\right\}$ for the set of all functions from $B$ to $A$. For the ease of notation 
we associate the elements in $\varphi\in\cA^{B}$, for arbitrary subsets $B\subseteq \cE$, with the vectors in $\cA^{|B|}$. This is made precise by fixing an ordering 
{\lq\lq}$<${\rq\rq} of the elements of $\cE$ and replacing $\varphi$ by the vector of images $\left(\varphi(b_1),\dots,\varphi(b_{|B|})\right)\in\cA^{|B|}$, where 
$\left\{b_1,\dots,b_{|B|}\right\}=B$ and $b_1<\dots<b_{|B|}$. For our network depicted in Figure~\ref{fig_generalized_diamond_network}, see Definition~\ref{def_family_B} 
for a formalization, we assume the order $e_0<\dots<e_{s+1}<f_0<\dots<f_s$ for the edges. E.g., if $B\subseteq \cE$ contains $e_0$, then the corresponding 
symbol on that edge is the first coordinate of each vector $x\in\cA^{|B|}$. 

\begin{definition}
  Let $\cN=(\cV,\cE,S,\mathbf{T})$ be a network and $\cA$ an alphabet. A \emph{network code} $F$ for $(\cN,\cA)$ is a 
  family of functions $\left\{F_V\,:\, V\in \cV\left(\backslash\{S\}\cup \mathbf{T}\right)\right\}$, where 
  $F_V\colon \cA^{|in(V)|}\to\cA^{|out(V)|}$ for all intermediate nodes $V$. 
\end{definition}

\begin{definition}
  An \emph{(outer) code} for a network $\cN=(\cV,\cE,S,\mathbf{T})$ and an alphabet $\cA$ is a subset 
  $C\subseteq \cA^{|out(S)|}$ with $|C|\ge 1$. 
\end{definition}

The elements of an outer code $C$ are called \emph{codewords}. 
For any two $c,c'\in \cA^{|B|}$, where $B\subseteq \cE$ is arbitrary, we define the \emph{Hamming distance} $d(c,c')$ as 
the number of differing coordinates, i.e., $d(c,c'):=\left\{1\le i\le |B|\,:\, c_i\neq c_i'\right\}$. Note that the Hamming distance is a metric. For each subset 
$C\subseteq\cA^{|out(S)|}$ we write $d(C)$ for the minimum Hamming distance $d(c,c')$ for any pair of different elements $c,c'\in C$, where we formally set $d(C):=\infty$ if $|C|\le 1$.    

In this paper we consider networks that are affected by potentially restricted adversarial noise. More precisely, we assume that at most 
$t$ of the alphabet symbols on a given edge set $\cU \subseteq \cE$ can be changed into any other alphabet symbols. 

\begin{definition}
  Let $\cN=(\cV,\cE,S,\mathbf{T})$ be a network, $\cA$ an alphabet, $T \in \mathbf{T}$ a terminal, $F$ a
  network code for $(\cN , \cA)$, $\cU \subseteq \cE$ a subset of the edges, and $t \ge 0$ an integer. We denote by
  $\Omega[\cN , \cA, F, S \to T, \cU, t] \colon \cA^{|out(S)|}\to \cA^{|in(T)|}$ the (adversarial) \emph{channel} representing the transfer 
  from $S$ to terminal $T \in \mathbf{T}$, when the network code $F$ is used by the vertices and at most $t$ edges 
  in $\cU$ are manipulated. In this context, we call $t$ the \emph{adversarial power}.
\end{definition}

Let $\Omega[\cN , \cA, F, S \to T, \cU, t]$ be an adversarial channel and $C\subseteq \cA^{|out(S)|}$ be an outer code for $\cN$. We say that $C$ is 
\emph{good} for $\Omega$ if $\Omega(c)\neq \Omega(c')$ for all codewords $c,c'\in C$ with $c\neq c'$. The \emph{($1$-shot) capacity} of $(\cN , \cA, \cU, t)$ 
is the largest real number $\kappa$ for which there exists an outer code $C\subseteq \cA^{|out(S)|}$ and a network code $F$ for $(\cN , \cA)$ with 
$\kappa = \log_{|\cA|}(|C|)$ such that $C$ is good for each adversarial channel $\Omega[\cN , \cA, F, S \to T, \cU, t]$ for all $T \in \mathbf{T}$. 
The notation for this largest $\kappa$ is $C_1(\cN , \cA, \cU, t)$.

\medskip

In the general model of Kschischang and Ravagnani, see \cite{ravagnani2018adversarial}, an adversary can modify $t$ edges from a given subset $\cU$ of all edges. 
Here we restrict ourselves to the situation where $\cU=out(S)$, i.e., the set of all outgoing edges of the unique source. Additionally, we assume that the adversary 
is omniscient, i.e., he can see all the sent symbols. With this,
given a network $\cN=(\cV,\cE,S,\mathbf{T})$ and an alphabet $\cA$ we say that a pair of an outer code $C\subseteq \cA^{|out(S)|}$ and a network code $F$ for $\cN$ is $t$-error correcting 
if $C$ is good for each adversarial channel $\Omega[\cN , \cA, F, S \to T, out(S), t]$ for all $T \in \mathbf{T}$. We also say that $C$ can be $t$-error 
corrected on $\cN$ if a network code $F$ exists such that $(C,F)$ is $t$-error correcting for $\cN$.  

\begin{lemma}
  \label{lemma_min_dist}
  Let $out(S)\subseteq \cU$. If an outer code $C$ can be $t$-error corrected for a network $\cN=(\cV,\cE,S,\mathbf{T})$ and an arbitrary alphabet $\cA$, then we have $d(C)\ge 2t+1$. 
\end{lemma}
\begin{proof}
  Assume that $d(C)\le 2t$ and let $c,c'$ be two codewords with $d(c,c')\le 2t$. Then there exists a vector $x\in \cA^{|out{S}|}$ with $d(c,x)\le t$ and $d(x,c')\le t$. 
  So, for each of the codewords $c$ and $c'$ the adversary can modify the values of the edges in $out(S)\subseteq\cU$ to $x$, so that $\Omega(c)=\Omega(c')$ while $c\neq c'$ -- 
  contradiction.
\end{proof}
If $\cU\not\supseteq out(S)$ then Lemma~\ref{lemma_min_dist} does not need to be valid in general.

In this paper we will consider a specific parameterized family of networks:
\begin{definition}
  \label{def_family_B}
  For each positive integer $s$ the network $\cN_s$ has one source $S$, one terminal 
  $T$, and two intermediate nodes $V_1$, $V_2$. There is exactly one directed edge from $S$ to $V_1$, one directed 
  edge from $V_1$ to $T$, $s+1$ parallel edges from $S$ to $V_2$, and $s$ parallel edges from $V_2$ to $T$,
  see Figure~\ref{fig_generalized_diamond_network}.  
\end{definition}

We remark that the notation $B_s$ was used in \cite{beemer2022network_arxiv} for $\cN_s$ and that $\cN_1=B_1$ is the 
original diamond network, see \cite{beemer2021curious}. In $\cN_s$ we label the edge from $S$ to $V_1$ as $e_0$ and the $s+1$ edges from $S$ to $V_2$ by $e_1,\dots,e_{s+1}$. 
The unique edge from $V_1$ to $T$ is denoted by $f_0$ and the $s$ edges from $V_2$ to $T$ are labeled by $f_1,\dots, f_s$,  
cf.~Figure~\ref{fig_generalized_diamond_network}. 

\begin{example}
  \label{ex_1}
  For alphabet $\cA=\{0,1,2\}$, $t=1$, and $s=2$ a code of size $5$ for $\cN_s$ is given by $C=\{0 0 0 0,0 1 1 1, 1 0 1 2,1 1 2 0,2 0 2 1\}$. We easily determine 
  $d(C)=3$ matching the lower bound in Lemma~\ref{lemma_min_dist}.
\end{example}

Since we are interested in fixed alphabet sizes in the following we will directly consider the maximum cardinality $\sigma(\cN_s,|\cA|)$ of an outer code 
$C$ that can be $1$-error corrected for $\left(\cN_s,\cA\right)$, i.e., we have $\log_{|\cA|} \sigma(\cN_s,|\cA|)=C_1(\cN_s,\cA,out(S),1)$. In Example~\ref{ex_2} we will 
continue Example~\ref{ex_1} and show that the stated code can be $1$-error corrected, i.e., we have $\sigma(\cN_s,3)\ge 5$. 

As an abbreviation we set
\begin{equation}
  \overline{C}:= \left\{ c'\in\cA^{s+2}\,\mid\, \exists c\in C: d(c,c')\le 1\right\}
\end{equation}
for the set of all possible manipulations of a codeword and the codewords themselves.

\section{Structure of network codes for generalized diamond networks $\cN_s$}
\label{sec_structure}

The aim of this section is to obtain some structure results on network codes $F$ that are $1$-error correcting for a given 
outer code $C$ on the network $\cN_s$ as well as on outer codes $C$ that can be $1$-error corrected.  
Lemma~\ref{lemma_min_dist} gives the lower bound $d(C)\ge 3$ for the minimum Hamming distance.



For each $c\in \cA^{s+2}$ we write $c=\left(c_1,c_2\right)$, where $c_1\in\cA$. Given our ordering of the edges, the symbol sent along 
edge $e_0$ (before the possible modification of the adversary) is $c_1$.
As an abbreviation we set 
\begin{equation}
  C_2:=\left\{c_2\,\mid\, \left(c_1,c_2\right)\in C\right\}.
\end{equation}  
Since $C_2$ is just a truncation of $C$ we have $d(C_2)\ge 2$ in our $1$-error correcting setting.

\begin{definition}
  \label{def_tau}
  For each $c_2\in\cA^{s+1}$ we define
  \begin{equation}
    \tau(c_2):=\left\{\left(a,c_2'\right)\in C\,:\,a\in\cA, d(c_2,c_2')\le 1\right\}.
  \end{equation}
\end{definition}
In words, if $F_{V_2}$ receives the input $c_2$, then the originally sent codeword $c'\in C$ has to be contained in $\tau(c_2)$. Moreover, the 
adversary can ensure that $F_{V_2}$ receives the input $c_2$ for all codewords $c'\in\tau(c_2)$. Clearly, we have $|\tau(c_2)|=1$ if $c_2\in C_2$.  


\begin{lemma}
  Let $C$ be a code that can be $1$-error corrected for the network $\cN_s$ and alphabet $\cA$. For each $c_2\in\cA^{s+1}$ we have $|\tau(c_2)|\le 
  \min\{|\cA|,s+1\}$.
\end{lemma}
\begin{proof}
  For $|\tau(c_2)|\le 1$ the statement is true due to our assumptions $s\ge 1$ and $|\cA|\ge 2$, so that we assume $|\tau(c_2)|\ge 2$ in the following.  
  Since $d(C_2)\ge 2$ and $d(c_2,c_2')\le 1$ for each codeword $(c_1',c_2')\in C$ with $c_1'\in\cA$ we have $d(c_2,c_2')=1$ (using $|\tau(c_2)|>1$). I.e.\ $c_2'$ 
  differs from $c_2$ in exactly one of the $s+1$ coordinates. Since no two codewords can differ in the same coordinate due to $d(C_2)\ge 2$ we have $|\tau(c_2)|\le s+1$. 
  Moreover, $d(C)\ge 3$ implies that all codewords in $\tau(c_2)$ pairwise differ in their first coordinate, so that $|\tau(c_2)|\le |\cA|$.
\end{proof}

\begin{lemma}
  \label{lemma_disjoint_e_0_components}
  Let $(C,F)$ be $1$-error correcting for $\cN_s$. For each pair of different elements $c_2,c_2'\in \cA^{s+1}$ with $F_{V_2}(c_2)=F_{V_2}(c_2')$ 
  we have that $\tau(c_2)\cup\tau(c_2')$ does not contain two different elements with coinciding first coordinate. 
\end{lemma}
\begin{proof}
  Let $c_2,c_2'\in \cA^{s+1}$ with $F_{V_2}(c_2)=F_{V_2}(c_2')$. Assume that there exist $a\in\cA$, $b,b'\in C_2$ with $b\neq b'$, and $(a,b),(a,b')\in \tau(c_2)\cup\tau(c_2')$.
  So, there exists an adversary channel $\Omega$ that changes the codeword $x:=(a,b)$ to either $(a,c_2)$ or $(a,c_2')$ and also the codeword $x':=(a,b')$ to either $(a,c_2)$ or $(a,c_2')$. 
  Using $F_{V_2}(c_2)=F_{V_2}(c_2')$ we compute $\Omega(x)=\left(F_{V_1}(a),F_{V_2}(c_2)\right)=\Omega(x')$ while $x\neq x'$ -- contradiction.
\end{proof}
 

\begin{lemma}
  \label{lemma_codewords_c_2}
  Let $(C,F)$ be $1$-error correcting for $\cN_s$. Let $c:=\left(c_1,c_2\right)\in C$ be a codeword and $c':=\left(c_1',c_2'\right)\in \overline{C}$ with $c_1,c_1'\in\cA$, $c\neq c'$ and 
  $F_{V_2}(c_2)= F_{V_2}(c_2')$. Then, we have $\tau(c_2')=\{c\}$.
\end{lemma}  
\begin{proof}
  Assume that there exists a codeword $\tilde{c}:=\left(\tilde{c}_1,\tilde{c}_2\right)\in\tau(c_2')$ with $\tilde{c}\neq c$ and $\tilde{c}_1\in\cA$. 
  So, there exists an adversary channel $\Omega$ that changes the codeword $c$ to $(\tilde{c}_1,c_2)$ and the codeword $\tilde{c}$ to $(\tilde{c}_1,c_2')$. 
  Using $F_{V_2}(c_2)=F_{V_2}(c_2')$ we compute $\Omega(c)=\left(F_{V_1}(\tilde{c}_1),F_{V_2}(c_2)\right)=\left(F_{V_1}(\tilde{c}_1),F_{V_2}(c_2')\right)=\Omega(\tilde{c})$ 
  while $c\neq \tilde{c}$ -- contradiction.
\end{proof}

Based on Lemma~\ref{lemma_disjoint_e_0_components} and Lemma~\ref{lemma_codewords_c_2} we can state a characterization of all outer codes that can be $1$-error corrected for $\cN_s$: 
\begin{theorem}  
  \label{thm_cover_problem}
  An outer code $\emptyset\neq C\subseteq \cA^{s+2}$ can be $1$-error corrected for $\cN_s$ iff there exists a set $\mathcal{B}$ 
  consisting of subsets of $C$ of cardinality $|\mathcal{B}|\le \left|\cA\right|^s-|C|$ such that the elements within each subset have pairwise different 
  first coordinates and for each $c_2\in \cA^{s+1}$ with $|\tau(c_2)|>1$ there exists an element $B\in\mathcal{B}$ with $\tau(c_2)\subseteq B$.  
\end{theorem}
\begin{proof}
  Let $(C,F)$ be $1$-error correcting for $\cN_s$ and define 
  \begin{equation}
    S_b=\cup_{c_2\in\cA^{s+1}\,:\, F_{v_2}(c_2)=b} \tau(c_2)\subseteq C 
  \end{equation}      
  for all $b\in \cA^s$. From Lemma~\ref{lemma_disjoint_e_0_components} we conclude that the elements in $S_b$ have pairwise different first coordinates for each $b\in \cA^s$. 
  For 
  \begin{equation}
    \mathcal{B}:=\left\{S_b\,:\, b\in \cA^s, \left|S_b\right|\ge 2 \right\}
  \end{equation}   
  we have that for each $c_2\in \cA^{s+1}$ with $|\tau(c_2)|>1$ there exists an element $B\in\mathcal{B}$ with $\tau(c_2)\subseteq B$. 
  For each codeword $c=\left(c_1,c_2\right)$ with $c_1\in\cA$ we have $S_{F_{V_2}(c_2)}=\{c\}$ due to Lemma~\ref{lemma_codewords_c_2}. So using Lemma~\ref{lemma_disjoint_e_0_components}, 
  at least $|C|$ of the sets $S_b$ have cardinality $1$ and $|\mathcal{B}|\le |\cA|^s-|C|$.
  
  \medskip
  
  For the other direction let a set $\mathcal{B}$ satisfying the stated properties be given. We set $\mathcal{B'}:=\big\{\{c\}\,:\,c\in C\big\}$ and $\mathcal{B}'':=\mathcal{B}\cup\mathcal{B}'$, so 
  that $|\mathcal{B}''|\le |\cA|^s$. Let $\pi\colon \mathcal{B}''\to\cA^s$ be an arbitrary injection. With this we define the network code $F$ for $\cN_s$ as follows. We choose 
  $F_{V_1}$ as the identity mapping, i.e., $F_{V_1}(x)=x$ for all $x\in\cA$. For each $c_2\in \cA^{s+1}$ with $|\tau(c_2)|> 1$ we set $F_{V_2}(c_2)$ to the lexicographically minimal element in 
  $\left\{\pi(B)\,:\, B\in\mathcal{B}, \tau(c_2)\subseteq B\right\}$. If $|\tau(c_2)|=1$ we set $F_{V_2}(c_2)$ to $\pi(B)$ for the unique element $B\in \tau(c_2)\in \mathcal{B}'$ and 
  to $\pi(B)$ for an arbitrary element $B\in\mathcal{B}'$ if $\tau(c_2)=\emptyset$.  
  With this $F$ is well defined. In order to show that $(C,F)$ is $1$-error correcting for $\cN_s$ we construct a {\lq\lq}decoder{\rq\rq} $F_T\colon\cA^{|in(T)|}\to\cA^{|out(S)|}$ with 
  $F_T(\Omega(c))=c$ for every codeword $c\in C$ and every adversary channel $\Omega$. For a given $y:=\left(y_1,y_2\right)\in\cA^{s+1}$ with $y_1\in\cA$ let $B\in\mathcal{B}''$ be the unique 
  element with $\pi(B)=y_2$ or $B=\emptyset$ if no such element $B$ exists. If $|B|=0$, then we set $F_T(y)$ to an arbitrary element in $\cA^{s+2}$ since this case cannot occur if a codeword 
  is sent across the network. If $|B|=1$ we set $F_T(y)$ to the unique codeword contained in $B$. If $|B|\ge 2$, then the adversary has not modified the value on edge $e_0$, so that 
  we can set $F_T(y)$ to the unique codeword in $B$ whose first coordinate equals $y_1=F_{V_1}(y_1)$.
\end{proof}

\begin{example}
  \label{ex_2}
  As in Example~\ref{ex_1} we consider the network $\cN_2$, the alphabet $\cA=\{0,1,2\}$, and the outer code $C$ given by the codewords 
  $c_1=(0,0,0,0)$, $c_2=(0,1,1,1)$, $c_3=(1,0,1,2)$, $c_4=(1,1,2,0)$, and $c_5=(2,0,2,1)$. It can be easily checked that 
  the different values of $\tau(c_2')$ for all $c_2'\in\cA^{s+1}$ are given by  
  $$
    \Big\{
    \emptyset,     
    \{c_1\}, \{c_2\}, \{c_3\}, \{c_4\}, \{c_5\}, \{c_2,c_3,c_5\}, \{c_1,c_3\}, \{c_1,c_4,c_5\}, \{c_2,c_4,c_5\}
    \Big\}.
  $$  
  For e.g.\ $c_2'=(0,1,1)$ we have $\tau(c_2')=\{c_2,c_3,c_5\}$. For the set $\mathcal{B}$ in Theorem~\ref{thm_cover_problem} we may choose 
  $$
    \mathcal{B}=\big\{\{c_2,c_3,c_5\}, \{c_1,c_3\}, \{c_1,c_4,c_5\}, \{c_2,c_4,c_5\}\big\}.
  $$
  Note that the elements of $\mathcal{B}$ are given by $\tau(c_2')$ with $|\tau(c_2')|>1$ and no two elements of $\mathcal{B}$ 
  can be joined since the elements have cardinality at most $|\cA|=3$. Thus, all $3^2=9=5+4$ different output vectors of the intermediate node $V_2$ have 
  to be used. As injection $\pi$ we may choose e.g.\  
  $\pi(\{c_1\})=(0,0)$, $\pi(\{c_2\})=(0,1)$, $\pi(\{c_3\})=(0,2)$, $\pi(\{c_4\})=(1,0)$, $\pi(\{c_5\})=(1,1)$, $\pi(\{c_2,c_3,c_5\})=(1,2)$, 
  $\pi(\{c_1,c_3\})=(2,0)$, $\pi(\{c_1,c_4,c_5\})=(2,1)$, and $\pi(\{c_2,c_4,c_5\})=(2,2)$. With this we e.g.\ have $F_{V_2}\big((0,1,2)\big)=(0,2)$ 
  and $F_{V_2}\big((0,1,1)\big)=(1,2)$.  
\end{example}  

\begin{example}
  \label{ex_3}
  Again we consider the network $\cN_2$ and the alphabet $\cA=\{0,1,2\}$. Now let the outer code $C$ be given by the codewords  
  $c_1=(0,0,0,0)$, $c_2=(0,1,1,1)$, $c_3=(0,2,2,2)$, $c_4=(1,0,1,2)$, and $c_5=(2,0,2,1)$. Here the different values of $\tau(c_2')$ for all $c_2'\in\cA^{s+1}$ are given by  
  $$
    \Big\{
    \emptyset,     
    \{c_1\}, \{c_2\}, \{c_3\}, \{c_4\}, \{c_5\}, \{c_1,c_4\}, \{c_1,c_5\}, \{c_2,c_4,c_5\}, \{c_3,c_4,c_5\}
    \Big\}.
  $$
  The sets $\{c_1,c_4\}$ and $\{c_1,c_5\}$ can be joined to $\{c_1,c_4,c_5\}$, so that we set $$\mathcal{B}=\big\{ \{c_1,c_4,c_5\}, \{c_2,c_4,c_5\}, \{c_3,c_4,c_5\}\big\}$$ in 
  Theorem~\ref{thm_cover_problem} and only $|C|+|\mathcal{B}|=8<3^2$ different output states are used in the corresponding function $F_{V_2}$.
\end{example}

As a byproduct of the proof of Theorem~\ref{thm_cover_problem} we see that it does not make a difference in terms of the capacity if the adversary can manipulate edge $e_0$ or 
$e_0$ and $f_0$. We may even remove the intermediate vertex $V_1$ and replace the two edges $e_0$ and $f_0$ by a direct edge between $S$ and $T$ that can be manipulated by the 
adversary. Going over the proofs of Lemmas \ref{lemma_min_dist}-\ref{lemma_codewords_c_2} we see that the number of outgoing edges of $V_2$ does not play a role. So, 
denoting the network where we replace the $s$ outgoing edges of $V_2$ to $T$ by $s'$ such edges by $\cN_{s,s'}$ we can prove a variant of Theorem~\ref{thm_cover_problem} 
where we replace the upper bound on the cardinality by $|\mathcal{B}|\le \left|\cA\right|^{s'}-|C|$. Cf.\ the code of Example~\ref{ex_3} in the situation where we replace the 
two outgoing edges of $V_2$ in $\cN_2$ by three outgoing edges, which however can only transmit symbols from the restricted alphabet $\{0,1\}$. i.e., a rather specific case 
of the more general situation where the allowed alphabets may differ on the different edges. 

In other words, the conditions of Theorem~\ref{thm_cover_problem} say that the elements of $\mathcal{B}$ cover all $\tau(c_2)$ with $c_2\in\cA^{s+1}$ and $|\tau(c_2)|>1$. The underlying 
cover problem can be solved by a straightforward integer linear programming (ILP) formulation, as it is usually the case for cover problems. This yields a first algorithm to determine the largest 
possible size $\sigma(\cN_s,|\cA|)$ of a code that can be $1$-error corrected for $\cN_s$ and alphabet $\cA$ for any parameter $s\in\N_{\ge 1}$. Due to Lemma~\ref{lemma_min_dist} we just have to 
loop over all codes with minimum Hamming distance at least $3$ and use Theorem~\ref{thm_cover_problem} to check whether $C$ can be $1$-error corrected for $\cN_s$ via an ILP computation.


We have applied this approach to the list of classified optimal\footnote{Here by an optimal code we understand a code of the maximum possible size given block length, alphabet size, 
and minimum distance.} binary one-error-correcting codes. These are denoted
as $(n,|C|,3)$ codes, where the length is given by $n=s+2$ in our situation and $|C|$ is the maximum possible code size.  For lengths $4,\dots, 11$ the codes were classified 
in~\cite{ostergard1999optimal}, for lengths $12,13$  in \cite{krotov2011optimal}, and for lengths $14,15$ in \cite{ostergard2009perfect}. The corresponding numbers are 
summarized in Table~\ref{table_binary_one_error_correcting_codes} and the codes can be downloaded at \url{http://pottonen.kapsi.fi/codes}.
\begin{table}
  \begin{center}
    \begin{tabular} {lrrrrrrrrrrrrr}
    \hline  
    $n$         \!\!\!&\!\!\!3\!\!\!&\!\!\!4\!\!\!&\!\!\!5\!\!\!&\!\!\!6\!\!\!&\!\!\! 7\!\!\!&\!\!\! 8\!\!\!&\!\!\! 9\!\!\!&\!\!\! 10\!\!\!&\!\!\!  11\!\!\!&\!\!\!    12\!\!\!&\!\!\!    13\!\!\!&\!\!\!   14\!\!\!&\!\!\!  15\!\\
    $\max |C|$ \!\!\!&\!\!\!2\!\!\!&\!\!\!2\!\!\!&\!\!\!4\!\!\!&\!\!\!8\!\!\!&\!\!\!16\!\!\!&\!\!\!20\!\!\!&\!\!\!40\!\!\!&\!\!\! 72\!\!\!&\!\!\! 144\!\!\!&\!\!\!   256\!\!\!&\!\!\!   512\!\!\!&\!\!\! 1024\!\!\!&\!\!\!2048\!\\
    number        \!\!\!&\!\!\!1\!\!\!&\!\!\!2\!\!\!&\!\!\!1\!\!\!&\!\!\!1\!\!\!&\!\!\! 1\!\!\!&\!\!\! 5\!\!\!&\!\!\! 1\!\!\!&\!\!\!562\!\!\!&\!\!\!7398\!\!\!&\!\!\!237610\!\!\!&\!\!\!117832\!\!\!&\!\!\!38408\!\!\!&\!\!\!5983\!\\
    \hline
    \end{tabular}
    \caption{Number of non-isomorphic optimal binary $(n,|C|,3)$ codes.}
    \label{table_binary_one_error_correcting_codes}
  \end{center}
\end{table}
Since those enumerations only consider one-error correcting codes up to symmetry we have to loop over the coordinates and decide which one should be the one corresponding to $e_0$. Interestingly 
enough, for all $n\in\{3,\dots,15\}\backslash\{3,7,15\}$ all choices for $e_0$ and all choices for $C$ lead to a code that can be $1$-error corrected for 
$\cN_{n-2}$. Note that for $n\in\{3,7,15\}$ we deal with the parameters of the Hamming codes. Here no choice of $C$ or the coordinate of $e_0$ leads to 
a code that can be $1$-error corrected for $\cN_{n-2}$. Thus, we have $\sigma(\cN_1,2)\le 1$, $\sigma(\cN_2,2)=2$, $\sigma(\cN_3,2)=4$, $\sigma(\cN_4,2)=8$, $\sigma(\cN_5,2)\le 15$, 
$\sigma(\cN_6,2)=20$, $\sigma(\cN_7,2)=40$, $\sigma(\cN_8,2)=72$, $\sigma(\cN_9,2)=144$, $\sigma(\cN_{10},2)=256$, $\sigma(\cN_{11},2)=512$, $\sigma(\cN_{12},2)=1024$, and $\sigma(\cN_{13},2)\le 2047$.   

In the next section we will determine 
$\sigma(\cN_5,2)=14$ and $\sigma(\cN_s,a)$ for other small parameters $s$ and $a\ge 3$. We remark  
that $\sigma(\cN_1,a)=a-1$ was determined in \cite{beemer2022network,beemer2021curious}.

\section{Exact values and lower bounds for the maximum code sizes}
\label{sec_ilp}
As mentioned in the previous section, $\sigma(\cN_1,a)=a-1$ was determined in \cite{beemer2022network,beemer2021curious}. Here an element $\tilde{a}\in\cA$ is chosen as a special symbol.
With this a suitable code $C$ is a threefold repetition code of the symbols in $\cA\backslash\left\{\tilde{a}\right\}$. If $|\tau(c_2)|>1$, then $V_2$ outputs $\tilde{a}$ with the 
corresponding set $B=C$. Here we will determine a few exact values and lower bounds for $\sigma(\cN_s,a)$ for small parameters. Upper bounds are discussed in Section~\ref{sec_bounds}. 

One of the most simple algorithmic approaches is to loop over all $|\cA|$-ary codes of length $s+2$ and minimum Hamming distance at least three with increasing sizes. For each
candidate code $C$ we can use Theorem~\ref{thm_cover_problem} and a small ILP computation to decide whether $C$ can be $1$-error corrected. For alphabet size $|\cA|=3$ and $s=2$ two codes 
that can be $1$-error corrected and have cardinality $5$ are given in Example~\ref{ex_2} and Example~\ref{ex_3}.  
By exhaustive search we have verified that no $3$-ary code with length $4$, minimum Hamming distance at least $3$, and size $6$ can be $1$-error corrected, so that $\sigma(\cN_2,3)=5$.  

In the following we will only state the codes $C$ since a suitable collection of sets $\mathcal{B}$ can be easily computed using Theorem~\ref{thm_cover_problem} and a small ILP 
computation. For $|\cA|=3$ and $s=3$ a code of size 
$14$ is given by 
\begin{eqnarray*}
  C&=&\{0 0 0 0 0,0 0 1 1 1,0 0 2 2 2,0 1 0 1 2, 0 1 1 2 0, 0 2 2 0 1, 1 0 0 2 1, \\ 
  && 1 1 1 0 2,1 1 2 1 1,1 2 0 1 0,2 0 2 1 0,2 1 0 0 1, 2 2 0 2 2,2 2 1 0 0\}.
\end{eqnarray*} 
For $|\cA|=3$ and $s=4$ a code of size $35$ is given by 
\begin{eqnarray*}
  C&=&\{0 0 0 0 0 0,2 1 2 0 0 0,1 0 2 0 2 1,1 2 2 2 1 0,0 1 2 1 2 2,1 2 0 1 2 1,2 0 2 2 0 1,2 1 1 2 2 0, 1 1 0 2 1 1,\\ 
  &&  1 2 1 0 2 2, 0 0 1 1 2 0, 2 0 2 1 1 0,0 0 2 2 1 2, 0 2 0 0 1 2,0 2 1 1 0 2,0 1 1 0 0 1, 1 1 1 1 0 0,2 0 0 1 2 2,\\ 
  &&  2 0 0 0 1 1, 1 2 1 2 0 1, 0 2 2 1 1 1, 1 0 2 1 0 2, 0 1 0 1 1 0, 2 1 1 1 1 2, 0 0 0 2 2 1, 1 0 1 1 1 1, 2 2 0 2 0 0, \\  
  &&  1 1 0 0 2 0, 2 0 1 0 0 2, 0 1 0 2 0 2, 2 1 0 1 0 1, 0 2 2 0 2 0, 2 2 1 0 1 0, 1 1 2 0 1 2, 2 2 2 2 2 2\}.
\end{eqnarray*}   
By exhaustive search we have verified $\sigma(\cN_3,3)=14$. For $|\cA|=4$ and $s=2$ a code of size $9$ is given by 
\begin{eqnarray*}
  C&=&\{0 0 0 0,0 1 1 1, 0 2 2 2,1 0 1 2,1 1 0 3,2 0 2 1,2 1 3 0,3 2 0 1,3 3 1 0\}.
\end{eqnarray*}   
For $|\cA|=4$ and $s=3$ a code of size $31$ is given by 
\begin{eqnarray*}
  C&=&\{00000, 32222, 21202, 31010, 31133, 22311, 20323, 01320, 03131, 30120, \\ 
  &&    11100, 00232, 12021, 23233, 11331, 32101, 22132, 13201, 30302, 20210, \\ 
  &&    13113, 23020, 02303, 33330, 01213, 10033, 02012, 21121, 33003, 12230, \\ 
  && 13322\}.
\end{eqnarray*}
\begin{example}
  \label{ex_5}
  For $|\cA|=5$ and $s=2$ a code of size $15$ is given by the codewords $c_0=0000$, $c_1=4341$, $c_2=2410$, $c_3=2143$, $c_4=3211$, $c_5=1421$, 
  $c_6=2232$, $c_7=4104$, $c_8=3024$, $c_9=1244$, $c_{10}=1012$, $c_{11}=3302$, $c_{12}=0433$, $c_{13}=4220$, and $c_{14}=1330$. A set 
  $\mathcal{B}$ for the construction of a suitable network code is given by
  \begin{eqnarray*}
    && \Big\{
      \{c_0,c_2,c_7,c_{10},c_{11}\}, 
      \{c_0,c_3,c_7,c_8,c_9\},
      \{c_0,c_6,c_{11},c_{13},c_{14}\},
      \{c_0,c_8,c_{10},c_{13}\},\\ 
    &&\{c_1,c_2,c_{11},c_{12},c_{14}\}, 
      \{c_1,c_3,c_4,c_9,c_{12}\},
      \{c_1,c_5,c_6,c_8,c_{12}\},\\  
    &&\{c_2,c_4,c_5,c_{13}\}, 
      \{c_4,c_6,c_{10},c_{12}\}, 
      \{c_6,c_8,c_9,c_{12},c_{13}\}  
    \Big\}.
  \end{eqnarray*}
\end{example}

\medskip

For $|\cA| =2$ we have that $|\tau(c_2)|>1$ is equivalent to $|\tau(c_2)|=2$ and the two codewords in $\tau(c_2)$ are at Hamming distance $2$. Moreover, for each pair of codewords 
$a,b\in C_2$ with $d(a,b)=2$ there are exactly two vectors $x\in\cA^{s+1}$ with $d(a,x)=d(x,b)=1$. This characterization allows to determine $\sigma(\cN_s,2)$ via the 
following ILP formulation. For each $v\in \{0,1\}^{s+2}$ we introduce binary variables $x_v$ and binary variables $y_{v'}$ for each $v'\in \{0,1\}^{s+1}$. The meaning of the
$x_v$'s is given by $x_v=1$ iff $v\in C$. The meaning of the $y_{v'}$'s is given by $y_{v'}=1$ iff $|\tau(v')|>1$. We minimize the code size 
$|C|=\sum_{v\in\{0,1\}^{s+2}} x_v$ subject to the constraints
\begin{equation}
  \sum_{v\in \{0,1\}^{s+2}}\ x_v \,+\, \frac{1}{2}\cdot \!\!\!\!\!\sum_{v'\in \{0,1\}^{s+1}} y_{v'} \,\le 2^{s}
\end{equation} 
modeling the condition from Theorem~\ref{thm_cover_problem},
\begin{equation}
  y_{m} \ge \sum_{i\in \cA} x_{(i,a)}+\sum_{j\in\cA} x_{(j,b)}-1
\end{equation}  
for all $a,b\in\{0,1\}^{s+1}$ with $d(a,b)=2$ and all $m\in \{0,1\}^{s+1}$ with $d(a,m)=d(m,b)=1$, which ensures 
that the pair of codewords $(i,a)$ and $(j,b)$ can only be taken if $y_m=1$, and
\begin{equation}
  \sum_{v\in\{0,1\}^{s+2}\,:\, v|_S=u} x_v \le 1
\end{equation}
for all $u\in \{0,1\}^s$ and all $S\subseteq \{1,\dots,s+2\}$ with $|S|=s$, which models $d(C)\ge 3$, where $v|_S$ is the restriction of the vector $v$ to the coordinates in $S$.

For $s=5$ a code of size $14$ is given by 
\begin{eqnarray*}
  C &=& \{0001101,0011110, 0100100, 0101011, 0110111, 0111000, 1000111, \\ 
    & & 1001000, 1010100, 1011011, 1100001, 1101110,1110010, 1111101\}
\end{eqnarray*} 
and using the above ILP we could verify that there is no such code of size $15$, so that $\sigma(\cN_5,2)=14$. We remark that the bound $|C|+|\mathcal{B}|=|\cA|^s$ is met with equality in our example. 

\section{Upper bounds for the maximum code sizes}
\label{sec_bounds}
In \cite{beemer2022network_arxiv} the two bounds $\sigma(\cN_s,a)\le a^s$ and $\sigma(\cN_s,a)\ge a^{s-1}$, if $a$ is sufficiently large, were shown. Both have an easy explanation. 
For the lower bound we may restrict to network codes where $V_1$ just forwards and the terminal $T$ ignores the output from $V_1$, i.e., the entire decoding is actually performed 
in intermediate node $V_2$. Let $A_q(n,d)$ denote the maximum number of vectors in a $q$-ary code of word length $n$ and with Hamming distance $d$, so that 
$\sigma(\cN_s,a)\ge A_a(s+1,3)$. Since for a fixed length and minimum distance MDS codes exist for all sufficiently large alphabet sizes $a$ this implies $\sigma(\cN_s,a)\ge a^{s-1}$. 
For other known lower bounds for $A_q(n,d)$ we refer to e.g.\ \cite{bogdanova2001error,bogdanova2001bounds,brouwer1998bounds} 
and especially the webpage 
\url{https://www.win.tue.nl/~aeb/}. For upper bounds we can use Lemma~\ref{lemma_min_dist} to conclude $\sigma(\cN_s,a)\le A_a(s+2,3)$, so that the Singleton bound gives 
$\sigma(\cN_s,a)\le a^s$ for all parameters $s$ and $a$. In Table~\ref{table_ub_codes} we summarize some known upper bounds for $A_a(s+2,3)$ from the mentioned website.          

\begin{table}[htp]
  \begin{center}
    \begin{tabular}{rrrrrrrrrrrrrr}
      \hline
      $s$     & 2 &  3 &  4 &   5 &   6 &  2 &  3 &   4 &   5 &  2 &   3 &   4 &    5 \\ 
      $a$     & 3 &  3 &  3 &   3 &   3 &  4 &  4 &   4 &   4 &  5 &   5 &   5 &    5 \\ 
      $\le$ & 9 & 18 & 38 & 111 & 333 & 16 & 64 & 176 & 596 & 25 & 125 & 625 & 2291 \\      
      \hline
    \end{tabular}
    \caption{Upper bounds for $A_a(s+2,3)$.}
    \label{table_ub_codes}
  \end{center}   
\end{table}

In \cite[Remark 5.10]{beemer2022network_arxiv} the {\lq\lq}conjecture{\rq\rq} $\sigma(\cN_s,a)=\frac{a^s+a}{2}-1$ was mentioned. For $(\cN,a)\in 
\big\{(\cN_1,\star),(\cN_2,2),(\cN_3,2),(\cN_4,2),(\cN_2,3),(\cN_3,3),(\cN_2,4)\big\}$ code sizes $\frac{a^s+a}{2}-1$ can indeed be attained.  
However, in \cite{roos1978some}\footnote{Cf.\ \url{https://mathscinet.ams.org/mathscinet-getitem?mr=484738}} the upper bound
$$
  A_a(n,3) \le \frac{a^n \cdot  (hn-x(a-x))}{(hn+x)\cdot (hn+x-a)},
$$   
where $h=a-1$,  $n \equiv x \pmod a$, and $1\le x\le a$, was concluded for $a>2$ and sufficiently large $n$ from 
Delsarte's linear programming method. 
Thus, we have $A_a(n,3) < \frac{a^n}{(a-2)n}$ for $a>2$ and sufficiently large $n$, i.e., if the conjecture is true, then $a$ 
has to be sufficiently large for a given parameter $s$. For $s=4$ and alphabet size $3$ we have $35 \le \sigma(\cN_4,3)\le A_3(6,3)=38<41$ 
and also $\sigma(\cN_s,2)<\frac{2^s+2}{2}-1$ for all $s\ge 5$.  
We remark that the code of Example~\ref{ex_5} shows $\sigma(\cN_2,5)\ge 15>\frac{5^2+5}{2}-1$.

\medskip

While Lemma~\ref{lemma_min_dist} implies $d(C_2)\ge 2$ only, we can use the number of ordered pairs of codewords 
of $C_2$ at distance $2$ 
\begin{equation}
  \Lambda:=\left|\left\{(a,b)\in C_2^2\,:\, d(a,b)=2\right\}\right| 
\end{equation}
to lower bound $|\mathcal{B}|$ in Theorem~\ref{thm_cover_problem} and to derive the following necessary criterion:
\begin{lemma}
  \label{lemma_bound_pairs_at_distance_2}
  Let $C\subseteq \cA^{s+2}$ be a code that can be $1$-error corrected for $\cN_s$. 
  Then, we have 
  $$
    |C|+ \frac{\Lambda}{|\cA|(|\cA|-1)} \le |\cA|^s.
  $$ 
\end{lemma} 
\begin{proof}
  Let $\mathcal{B}$ as in Theorem~\ref{thm_cover_problem}. If $a,b\in C_2$ are two codewords with $d(a,b)=2$, then there exists a vector 
  $x\in \cA^{s+1}$ with $d(a,x)=d(x,b)=1$, so that $a,b\in \tau(x)$. Thus, there exists a set $B\in\mathcal{B}$ 
  with $\tau(x)\subseteq B$. Since $|B|\le |\cA$| at most $|\cA|(|\cA|-1)$ ordered pairs counted in $\Lambda$ can yield the same 
  set $B\in \mathcal{B}$, so that $|\mathcal{B}|\ge \frac{\Lambda}{|\cA|(|\cA|-1)}$ and the stated inequality follows from 
  Theorem~\ref{thm_cover_problem}.
\end{proof}

A first, rather weak, lower bound for $\Lambda$ can be obtained easily. For network $\cN_s$ the Singleton bound implies $|C_2|\le |\cA|^{s-1}$ if 
$d(C_2)\ge 3$, so that $C_2$ contains at least one pair of codewords at Hamming distance $2$. Iteratively removing a codeword in such a pair, we 
can conclude $\Lambda\ge |C|-|\cA|^{s-1}$ for $|C|\ge |\cA|^{s-1}$. 
In Subsection~\ref{subsec_linear_programming_method} we will utilize the linear programming method for codes to 
obtain lower bounds for $\Lambda$ and state the corresponding upper bounds for $\sigma(\cN_s,a)$ for small parameters.

\subsection{Upper bounds for the code sizes for $\cN_s$ based on the linear programming method}
\label{subsec_linear_programming_method}
We can also use Delsarte's linear programming method to derive lower bounds for $\Lambda$ and then apply Lemma~\ref{lemma_bound_pairs_at_distance_2} to upper bound the code 
size $|C|$. To this end we mention that for integers $n\ge 1$ and $q\ge 2$ the \emph{Krawtchouk polynomials} are defined as 
\begin{equation}
  K_i^{(n,q)}(z):=\sum_{j=0}^i (-1)^j (q-1)^{i-j} q^j {{n-j}\choose {n-i}}{z\choose j}
\end{equation}
for all $i\ge 0$, where ${z\choose j}:=z(z-1)\cdots (z-j+1)/j!$ for all $z\in\mathbb{R}$.  The vector $B(C) = \left(B_0, B_1,\dots, B_n\right)$, 
where
\begin{equation}
  B_i = \frac{1}{|C|} \cdot\left|\left\{(a,b)\in C^2\,\mid\, d(a,b)=i\right\}\right|, \,\,i=0,1,\dots,n 
\end{equation}
is called the \emph{distance distribution} of $C$. Clearly, $B_0 = 1$ and $B_i = 0$ for all $1\le i\le d(C)-1$. Moreover, $\Lambda=|C|\cdot B_2$.   
The vector $B'(C)=\left(B_0',B_1',\dots,B_n'\right)$, where
\begin{equation}
  B_i'=\frac{1}{|C|} \sum_{j=0}^n B_jK_i^{(n,q)}(j), \,\,i=0,1,\dots,n
\end{equation}
is called the \emph{dual distance distribution} of $C$. 
Obviously we have $B_0'=1$. 

\begin{theorem} (\cite{delsarte1973algebraic,delsarte1973four}) 
  The dual distance distribution of $C$ satisfies $B_i'\ge 0$ for all $0\le i\le n$.
\end{theorem}  
In Table~\ref{table_ub_codes_2} we have listed a few explicit upper bounds for $\sigma(\cN_s,a)$, where $2\le s\le 5$ and $3\le a\le 5$,  
based on Lemma~\ref{lemma_bound_pairs_at_distance_2} and the linear programming method (using $B_1=0$ and 
minimizing $B_2$). For more details we refer to the recent survey \cite{boyvalenkov2021linear}.

\begin{table}[htp]
  \begin{center}
    \begin{tabular}{rrrrrrrrrrrrrr}
      \hline
      $s$     & 2 &  3 &  4 &   5 &  2 &  3 &   4 &   5 &  2 &  3 &   4 &    5 \\ 
      $a$     & 3 &  3 &  3 &   3 &  4 &  4 &   4 &   4 &  5 &  5 &   5 &    5 \\ 
      $\le$ & 6 & 15 & 42 & 108 & 11 & 37 & 133 & 484 & 17 & 76 & 337 & 1512 \\      
      \hline
    \end{tabular}
    \caption{Upper bounds for $\sigma(\cN_s,a)$ based on the linear programming method.}
    \label{table_ub_codes_2}
  \end{center}   
\end{table}

We remark that the upper bounds in Table~\ref{table_ub_codes_2} seem to be better than those in Table~\ref{table_ub_codes} if 
$s$ is rather small. E.g.\ the approach from this subsection implies $\sigma(\cN_4,3)\le 42$, while we have $\sigma(\cN_4,3)\le A_3(6,3)=38$. 
On the other hand we have $\sigma(\cN_3,3)\le A_3(5,3)=18$ while Table~\ref{table_ub_codes_2} states $\sigma(\cN_3,3)\le 15$ noting that 
we have determined $\sigma(\cN_3,3)=14$ by exhaustive enumeration in Section~\ref{sec_ilp}. We give a summary of the best known bounds 
for small parameters in the subsequent conclusion. 

\section{Conclusion}
\label{sec_conclusion}

In this paper we have considered the problem of determining the maximum size $\sigma(\cN_s,a)$ of a code that can be $1$-error corrected 
for the network $\cN_s$ and a given alphabet of size $a$. The considered family of networks is one of five infinite 
families of so-called $2$-level networks studied in \cite{beemer2022network_arxiv}, more precisely family~B. Bounds for such $2$-level networks 
can be used to conclude upper bounds for more general networks as demonstrated in \cite{beemer2022network_arxiv}. Besides that they form a 
starting point for first exact computations there is some hope that the obtained insights might be generalized. In this context we mention that determining the
maximum size of a code that can be $t$-error corrected for the most simple {\lq\lq}network{\rq\rq} consisting just of a source $S$, a terminal $T$, 
and $n$ parallel edges between $S$ and $T$ over an alphabet of size $q$ corresponds to the determination of $A_q(n,2t+1)$, i.e., the maximum number 
of vectors in a $q$-ary code of word length $n$ and with Hamming distance $2t+1$, which triggered a lot of research and turned out to be a quite intricate 
problem that does not seem to admit an easy closed-form solution.

Our main result is a characterization result for outer codes that can be $1$-error corrected for network $\cN_s$ in Theorem~\ref{thm_cover_problem}. 
On the computational side the problem of deciding whether a given code $C$ can be $1$-error corrected (and the construction of an achieving network code) 
is reduced to a covering problem. Via exhaustive enumeration of codes we have determined a few exact values of $\sigma(\cN_s,a)$ in Section~\ref{sec_structure} 
and Section~\ref{sec_ilp}. For the special case of a binary alphabet we also stated a direct integer linear programming formulation 
for $\sigma(\cN_s,2)$ which was used to compute $\sigma(\cN_5,2)=14$. With respect to upper bounds for the code size $|C|$ the easy observation $d(C)\ge 3$ 
for the minimum Hamming distance yields $\sigma(\cN_s,a)\le A_a(s+2,3)$, so that known results from the literature can be applied. We e.g.\ remark 
that the {\lq\lq}conjecture{\rq\rq} $\sigma(\cN_s,a)=\frac{a^s+a}{2}-1$ from \cite[Remark 5.10]{beemer2022network_arxiv} is wrong if one does not assume that 
the alphabet size $a$ is sufficiently large. Also based on the characterization result in Theorem~\ref{thm_cover_problem} we have utilized 
the number of codewords of Hamming distance $2$ to upper bound the maximum possible code size, see Lemma~\ref{lemma_bound_pairs_at_distance_2} in Section~\ref{sec_bounds}. 
An application of the linear programming method for codes 
yields numerical improvements on $\sigma(\cN_s,a)<a^s$, see Subsection~\ref{subsec_linear_programming_method}. We summarize our numerical results for $\sigma(\cN_s,a)$ 
and small parameters in Table~\ref{table_summary}.       

\begin{table}[htp]
  \begin{center}
    \begin{tabular}{rrrrrrrrrrrrrr}
      \hline
      $s$               & 1 & 2 & 3 & 4 &  5 &  6 &  7 &  8 &   9 &  10 &  11 &   12 \\ 
      $a$               & 2 & 2 & 2 & 2 &  2 &  2 &  2 &  2 &   2 &   2 &   2 &    2 \\ 
      $\sigma(\cN_s,a)$ & 1 & 2 & 4 & 8 & 14 & 20 & 40 & 72 & 144 & 256 & 512 & 1024 \\      
      \hline
    \end{tabular}
    \smallskip
    \begin{tabular}{rrrrrrrrrrrrr}
      \hline
      $s$               & 1 & 2 &  3 &      4 & 1 &     2 &      3 & 1 &      2 \\ 
      $a$               & 3 & 3 &  3 &      3 & 4 &     4 &      4 & 5 &      5 \\ 
      $\sigma(\cN_s,a)$ & 2 & 5 & 14 & 35--38 & 3 & 9--11 & 31--37 & 4 & 15--17 \\      
      \hline
    \end{tabular}
    \caption{Best known bounds for $\sigma(\cN_s,a)$ and small parameters.}
    \label{table_summary}
  \end{center}   
\end{table}

We propose the determination of tighter bounds for $\sigma(\cN_s,a)$ as an interesting open problem. 
In our opinion also the two infinite families $A$ and $E$ from \cite{beemer2022network_arxiv} deserve a 
tailored study for fixed alphabet sizes. For family $E$ one can use a subset of a repetition code and 
a suitable network code to achieve code sizes that grow linearly in the alphabet size. 
For network $A_t$, where the adversary can manipulate $t$ of the outgoing 
edges of the source $S$, a similar result as Theorem~\ref{thm_cover_problem} can be shown. Without going into 
details we just state that one can assume that intermediate vertex $V_1$ forwards the information from its $t$ incoming 
edges and for intermediate vertex $V_2$ we can define 
\begin{eqnarray*}
  \tau(c_2)&:=&\Big(\left\{c'\in C\,\mid\, \exists c_1\in \cA^t; d(c',(c_1,c_2))=1\right\},\dots,\\ 
  && \left\{c'\in C\,\mid\, \exists c_1\in \cA^t; d(c',(c_1,c_2))=t\right\}\Big)
\end{eqnarray*} 
for all $c_2\in\cA^{2t}$. The {\lq\lq}analog{\rq\rq} of $\mathcal{B}$ would be a set whose elements consist of lists of length $t$ of 
subsets of the outer code $C$. The interpretation is that those lists contain candidates for the originally sent codeword 
assuming that a certain number of incoming edges of $V_2$ was attacked. Together with the information from $V_1$ 
a unique decoding should be possible. Of course everything gets more technical and we currently do not know for which classes 
of 2-level networks a similar approach might work or what a good generalization would be.  

\section*{Acknowledgements}
The author  would like to thank the organizers of the Sixth Irsee Conference on Finite Geometries for their invitation. 
Further thanks go to Altan K{\i}l{\i}{\c{c}} who catched the authors interest by talks on the mentioned conference and 
the 25th International Symposium on Mathematical Theory of Networks and Systems (MTNS) 2022 as well as several discussions 
during the coffee breaks and afterwards. Last but not least the author thanks the two anonymous referees for their careful 
reading and their valuable comments that helped to improve the presentation of this paper.   


\end{document}